\documentclass[conference,letterpaper]{IEEEtran}
\usepackage{amssymb}
\usepackage{amsthm}
\usepackage{placeins}

\usepackage[dvipsnames]{xcolor}
\definecolor{DarkGreen}{rgb}{0.1,0.5,0.1}
\definecolor{DarkRed}{rgb}{0.5,0.1,0.1}
\definecolor{DarkBlue}{rgb}{0.1,0.1,0.5}

\usepackage{mleftright}

\usepackage[small]{caption}
\usepackage[pdftex]{hyperref}
\hypersetup{
    unicode=false,          
    pdftoolbar=true,        
    pdfmenubar=true,        
    pdffitwindow=false,      
    pdfnewwindow=true,      
    colorlinks=true,       
    linkcolor=DarkBlue,          
    citecolor=DarkGreen,        
    filecolor=DarkGreen,      
    urlcolor=DarkBlue,          
    pdftitle={},
    pdfauthor={},    
}
\usepackage{amssymb,amsthm,amsfonts,enumitem, dsfont}
\usepackage{comment}
\usepackage{tikz}
\usepackage{pgfplots}
\pgfplotsset{compat=1.14}
\usetikzlibrary{arrows,decorations.pathmorphing,decorations.shapes,snakes,patterns,positioning}
\usepackage{thm-restate}

\newcommand{\N}{\ensuremath{\mathbb{N}}}

\newcommand{\EE}{\mathbb{E}}

\newcommand{\argmax}{\mathrm{argmax}}

\renewcommand{\epsilon}{\varepsilon}

\definecolor{brightpink}{rgb}{1.0, 0.0, 0.5}
\definecolor{byzantine}{rgb}{0.74, 0.2, 0.64}
\definecolor{byzantium}{rgb}{0.44, 0.16, 0.39}

\theoremstyle{plain}
\declaretheorem[name=Theorem,numberwithin=section]{theorem}
\declaretheorem[name=Lemma,sibling=theorem]{lemma}
\newtheorem*{lemma*}{Lemma} 
\newtheorem*{theorem*}{Theorem} 
\newtheorem{definition}[theorem]{Definition}

\DeclareFontFamily{U}{mathx}{\hyphenchar\font45}
\DeclareFontShape{U}{mathx}{m}{n}{<-> mathx10}{}
\DeclareSymbolFont{mathx}{U}{mathx}{m}{n}
\DeclareMathAccent{\widebar}{0}{mathx}{"73}

\usepackage{algorithm}

\usepackage[utf8]{inputenc}

\usepackage[T1]{fontenc}
\usepackage{url}
\usepackage{ifthen}
\usepackage{cite}
\usepackage[cmex10]{amsmath} 

\begin{document}
\title{Real-time oblivious erasure correction with linear time decoding and constant feedback} 


\author{%
  \IEEEauthorblockN{Shashwat Silas}
  \IEEEauthorblockA{Stanford University\\
                    353 Serra Mall,\\
                    CA 94305, USA\\
                    Email: silas@stanford.edu}

}

\maketitle

\begin{abstract}
We continue the study of rateless codes for transmission of information across channels whose rate of erasure is unknown. In such a code, an infinite stream of encoding symbols can be generated  from the message and sent across the erasure channel, and the decoder can decode the message after it has successfully collected a certain number of encoding symbols. A rateless erasure code is real-time oblivious if rather than collecting encoding symbols as they are received, the receiver either  immediately decodes or discards each symbol it receives. 
  
  Efficient real-time oblivious erasure correction uses a feedback channel in order to maximize the probability that a received encoding symbol is decoded rather than discarded. We construct codes which are real-time oblivious, but require fewer feedback messages and have faster decoding compared to previous work in \cite{beimel2007rt}. Specifically, for a message of length $k'$, we improve the expected complexity of the feedback channel from $O(\sqrt{k'})$ to $O(1)$, and the expected decoding complexity from $O(k'\log(k'))$ to $O(k')$. Our method involves using an appropriate block erasure code to first encode the $k'$ message symbols, and then using a truncated version of the real-time oblivious erasure correction of \cite{beimel2007rt} to transmit the encoded message to the receiver, which then uses the decoding algorithm for the outer code to recover the message. 
    \end{abstract}


\section{Introduction}

Reliable transmission of information on the internet has necessitated the study of new and efficient methods of erasure correction. Information on the internet is transmitted as packets, some of which may be lost or corrupted by the channel during transmission. This process of information loss is similar to the binary erasure channel (BEC) and has made the study practical erasure codes for the BEC important. Common techniques for achieving reliable transmission include using feedback from the receiver to find which packets have gone missing, as in Automatic Repeat Request (ARQ) protocols \cite{tanenbaum1996computer}, or encoding information with erasure codes \cite{luby2001efficient} which can be decoded by the receiver. 

As noted in \cite{beimel2007rt}, both these techniques have their drawbacks. ARQs can be inefficient when there is a time delay in sending messages and acknowledgements back and forth, and as such may not be suitable in settings where there are time constraints on the transmission of information. They also require the extensive use of feedback from the receiver, whereas in many situations communication channels are asymmetric. Erasure codes require additional processing by the receiver (decoding) and an understanding of the rate of information loss over the channel for efficiency. They also do not make use of any feedback from the receiver, whereas at least a lean feedback channel is often present.

Innovations in rateless codes such as \cite{luby2002lt} and \cite{shokrollahi2006raptor} have created efficient methods for transmitting information using erasure codes which require no knowledge of the rate of information loss of the channel. In particular, $k$ message symbols can be transmitted as an infinite stream of encoding symbols, and a decoder which collects \emph{any} (slightly more than) $k$ encoding symbols can successfully decode the original message. In these methods, the decoder needs to wait until they collect a certain number of encoding symbols to begin decoding and the computational effort of decoding is concentrated towards the end of the transmission (sometimes also called the "waterfall" phenomenon). The real-time oblivious erasure correction protocol of \cite{beimel2007rt} removes this restriction. In their protocol encoding symbols are not stored, rather each received encoding symbol is either used immediately or discarded. This \emph{real-time} oblivious decoding reduces the memory overhead at the receiver and also makes the computational effort of decoding uniform over time throughout the transmission process. To accomplish this efficiently, the protocol of \cite{beimel2007rt} utilizes a feedback channel from the receiver to the sender. 

Several other works such as \cite{hagedorn2009rateless}, \cite{talari2013lt}, \cite{hashemi2016fountain}, \cite{hashemi2014near} have also studied the use of feedback channels to construct rateless codes with desirable properties. Different types of feedback messages have been studied, such as providing feedback on the number of symbols which have been decoded by the receiver such as in \cite{hagedorn2009rateless},\cite{beimel2007rt}, and distance type feedback messages in \cite{hashemi2014near}. While these works consider a variety of different properties, they all use a feedback channel and minimize the number of feedback messages needed to accomplish their goals. In this work, we focus  on the real-time oblivious erasure decoding model of \cite{beimel2007rt}, and we show how to significantly reduce the number of feedback messages required in their setting. Our techniques may be of interest for reducing the complexity of other feedback based rateless coding schemes. 

Specifically, we improve the parameters of the real-time oblivious erasure correction protocol of \cite{beimel2007rt}, which to the best of our knowledge are the state of the art real-time oblivious erasure correction codes. Specifically, we show that by using a similar pre-coding scheme to Raptor codes from \cite{shokrollahi2006raptor}, we can reduce the expected complexity of the feedback channel from $O(\sqrt{k'})$ to $O(1)$, and the expected decoding complexity from $O(k'\log(k'))$ to $O(k')$ for a message of length $k'$. The expected number of encoding symbols required for decoding is close to $2k'$ as in the original scheme, and we require a small memory overhead at the sender.

We point out that the codes of \cite{shokrollahi2006raptor} and \cite{beimel2007rt} rely on simple exclusive-or operations in order to encode and decode symbols and are practically efficient. Our codes, which are based on combining components of these works, also only rely on exclusive-or operations. 

\subsection{Real-time oblivious erasure correction}

In real-time oblivious erasure correction, a sender transmits a message of $k$ bits to a receiver. The transmission from the sender occurs in a rateless manner, in which an infinite stream of symbols which are called \emph{encoding symbols} is sent to the receiver over a BEC whose erasure rate is unknown. The receiver has access to the following memory: in $k$ bits it stores whether the $i$th message bit for $0\le i\le k$ has been successfully received, and in another $k$ it stores the values of any successfully received bits. It has access to additional memory in which it can store and process a \emph{single} encoding symbol. Every encoding symbol which is received is processed, and processing an encoding symbol results in either in a message symbol being decoded from the encoding symbol, or the encoding symbol being discarded. If it can be used to decode the $i$th symbol of the message for some $1 \le i \le k$, then the $i$th message symbol is marked as decoded in the memory, the decoded message bit is stored in appropriately, and the memory for processing encoding symbols is cleared. If the received encoding symbol cannot be decoded, then it is discarded from the processing memory immediately. The receiver also has access to a feedback channel through which it can transmit single bits to the sender at any point during the transmission. The efficiency of real-time oblivious erasure correction is measured by the following parameters.
\begin{itemize}
\item Processed symbols: the number of received symbols which are processed by the receiver before the transmission is complete.
\item Feedback complexity: the number of bits which need to be received by the sender through the feedback channel over the entire protocol.
\item Computation/Decoding complexity: the number of single-bit operations which are performed by the receiver in order to finish decoding the message. 
\end{itemize}
Note that these are rateless codes, and are agnostic to the erasure rate of the transmission channel or the feedback channel. All probabilistic statements about the transmission channel are made over the randomness of the encoding map, conditioned on which encoding symbols are received. Similarly, feedback complexity is measured by the number of bits the sender needs to receive, regardless of the feedback channel. 

\subsection{The real-time oblivious erasure correction of \cite{beimel2007rt}}\label{sec:protocol}

We first quickly recall the real-time oblivious erasure correction protocol of \cite{beimel2007rt}. In that work, $k$ symbols are transmitted over an erasure channel in the following way. The sender chooses an appropriate degree $d$, then chooses $d$ symbols of the message, $m_1,\dots,m_d$, uniformly at random without replacement, and transmits the exclusive-or of these messages symbols, $e = \oplus_{1}^{d}m_i$, across the channel as an encoding symbol $e$. The indices of the message symbols that have been included in the encoding symbol are also communicated to the receiver (one may refer to \cite{beimel2007rt} or \cite{luby2002lt} for details of this step). These encoding symbols can be transmitted infinitely in a rateless manner. The receiver \emph{processes} each encoding symbol it receives in the following way. If the receiver already knows exactly $d-1$ of the message symbols included in an encoding symbol, $m_{i_1}, \dots, m_{i_{d-1}}$ then it decodes the encoding symbol by computing $e \oplus m_{i_1} \oplus \dots  \oplus m_{i_{d-1}}$ and learns another message symbol. Otherwise, it discards the encoding symbol immediately. 

As a sanity check, one may note that if $r$ is the number of message symbols already decoded by the receiver, then no encoding symbols of degree $d \ge r+2$ can be decoded. Also note that if $r=0$, then any symbol of degree $d=1$ is decoded with probability $1$. Similarly if $r=k-1$, then the unique symbol of degree $k$ is decoded with probability $1$. In this way, if the value of $r$ is known to the sender, it may optimize $d$ in order to maximize the probability that an encoding symbol is decoded rather than discarded.

In their work, \cite{beimel2007rt} show that the probability of decoding an encoding symbol when $r$ message symbols are known to the decoder is maximized by choosing $d$ as a function of $r$ with $ d(r) = \left\lfloor\frac{k+1}{k-r}\right\rfloor$ for $0\le r \le k-2$ and $d(r)=k$ for $r = k-1$. Every time the value of $d(r)$ changes as the result of the receiver decoding another message symbol, the receiver sends an acknowledgement to the sender to update the encoding degree. 

They are able to show that in expectation: transmitting a message of $k$ symbols requires the receiver to process $2k$ encoding symbols, send $O(\sqrt{k})$ feedback messages, and the expected computational cost for processing all the received symbols is $O(k\log(k))$. 

\subsection{Our protocol}\label{sec:prot}
We modify this protocol in the following way. Similar to Raptor codes, we first encode the message using an appropriate erasure code that has linear time encoding and decoding algorithms. We then use (a truncated version of) the real-time oblivious transmission process of \cite{beimel2007rt} to transmit this codeword to the receiver. The receiver recovers a constant fraction of the symbols of the message (encoded with the outer code), and then relies on the outer code's erasure correction algorithm to decode the original message. In our protocol, in expectation,  transmitting a message of $k'$ symbols requires the receiver to process slightly more than $2k'$ encoding symbols, send $O(1)$ feedback messages, and the expected computational cost for processing all the received symbols is $O(k')$. 

\begin{definition}[The code $(k',\mathcal{R}_{k},1-\gamma)$] 
We use the notation $(k',\mathcal{R}_{k},1-\gamma)$ to refer to the following composed code.
\begin{itemize} 
\item A message $x$ of $k'$ symbols is first encoded with an outer erasure code $\mathcal{R}_{k} : \mathbb{F}_2^{k'} \to \mathbb{F}_2^{k}$ to produce a codeword $\mathcal{R}_{k}(x)$ of length $k$.
\item This codeword is then transmitted to the receiver using the real-time erasure correcting protocol of \cite{beimel2007rt}. Unlike their protocol, the receiver terminates the transmission after recovering a $(1-\gamma)$ fraction of the $k$ transmitted symbols.
\end{itemize}
\end{definition}

\begin{definition}[Decoding $(k',\mathcal{R}_{k},1-\gamma)$] 
The decoding of $(k',\mathcal{R}_{k},1-\gamma)$ occurs in two steps. First, the receiver recovers a $(1-\gamma)$ fraction of the $k$ transmitted symbols using the protocol of \cite{beimel2007rt}. Then the receiver uses the decoding algorithm of $\mathcal{R}_{k}$ in order to recover $x$. The overall decoding fails if the decoding of $\mathcal{R}_{k}$ fails. 
\end{definition}
 The code $(k',\mathcal{R}_{k},1-\gamma)$ is encoded and decoded using the following algorithms.

	\begin{algorithm}[H]

			\begin{itemize}
  			 \item Set $r =0$. Encode the message $x$ using the outer code $\mathcal{R}_{k}$ as $\{x_1,\dots,x_{k}\}$. Proceed to Algorithm 2
			\end{itemize}

		\caption{Step 1 for the encoder}
		\label{}
	\end{algorithm}

	\begin{algorithm}[H]

			\begin{itemize}
			    \item If a termination message is received, stop transmission. Otherwise,
			    \item If an acknowledgement is received from the receiver, update $r \leftarrow r + 1$
			    \item Choose the encoding degree $$d = \left\lfloor\frac{k+1}{k-r}\right\rfloor$$
			    \item Choose $d$ indices $\{i_1,\dots, i_d\}$ independently and uniformly at random from $\{1,\dots,k\}$
			    \item Compute $e = \oplus_{i=1}^{d}x_i$
			    \item Send $(e, \{i_1,\dots, i_d\}, d)$
			    
			\end{itemize}

		\caption{Step 2 for the encoder (Loop)}
		\label{}
	\end{algorithm}
	
		\begin{algorithm}[H]

			\begin{itemize}
			    
	\item Initialize the set of known symbols to be empty and $r = 0$. Proceed to Algorithm \ref{alg:receiver2}

			  \end{itemize}
	
		\caption{Step 1 for the decoder}
		\label{alg:receiver1}
	\end{algorithm}

	\begin{algorithm}[H]

			\begin{itemize}
			    			    \item If the number of known symbols is $(1-\gamma)k$, send a termination message to the encoder and proceed to Algorithm \ref{alg:receiver3}. Otherwise, receive $(e,\{i_1,\dots,i_d\},d)$.
			    \item If exactly one of $\{i_1,\dots, i_d\}$, say $i_j$, is unknown, then decode $e$ by computing the xor with the known symbols. Then update the set of known symbols to include $i_j$ and set $r \leftarrow r + 1$.
			    \item Discard $(e,\{i_1,\dots,i_d\},d)$.
			    \item If $d \not= \left\lfloor\frac{k+1}{k-r}\right\rfloor$, send an acknowledgement to encoder.
			    
			\end{itemize}

		\caption{Step 2 for the decoder (Loop)}
		\label{alg:receiver2}
	\end{algorithm}
	
	\begin{algorithm}[H]

			\begin{itemize}
			    
			    \item Use the decoding algorithm for $\mathcal{R}_{k}$ to recover $\mathcal{R}_{k}(x)$. If the decoding fails, then output FAIL.
			    			\end{itemize}

		\caption{Step 3 for the decoder}
		\label{alg:receiver3}
	\end{algorithm}

We remark that it is not necessary to include the set of indices $\{i_1,\dots,i_d\}$ along with every encoding symbol, but it makes the exposition easier to do so. Specific schemes for avoiding this have been provided in previous works on rateless codes such as \cite{beimel2007rt} and \cite{luby2002lt}. Since the encoder in our protocol also stores $\mathcal{R}_k(x)$, we define 
 \begin{definition}[Memory overhead of $(k',\mathcal{R}_{k},1-\gamma)$] 
Pre-coding the message $x$ with $\mathcal{R}_{k}$ requires the sender to store $\mathcal{R}_{k}(x)$ as well as $x$. We refer to these additional bits of $\mathcal{R}_{k}(x)$ as the memory overhead of $(k',\mathcal{R}_{k},1-\gamma)$. 
\end{definition}

\section{Results}
Compared to the work of \cite{beimel2007rt}, our protocol reduces the complexity of the feedback channel from $O(\sqrt{k'})$ to $O(1)$, and the expected decoding complexity from $O(k'\log(k'))$ to $O(k')$ for a message of length $k'$. To accomplish this, we use an appropriate block code to pre-code the original message as described in section \ref{sec:prot}. We first explain why pre-coding results in improvements over the original protocol. 

Suppose $k$ bits are being transmitted using the real-time oblivious erasure correction protocol of \cite{beimel2007rt}. In their work, the work of the feedback channel is non-uniform, and heavily concentrated towards the end of the transmission process. In \cite{beimel2007rt}, it is shown that by the time $k-\sqrt{k+1}$ symbols have been decoded, the feedback channel has already sent $\sqrt{k+1}$ messages. We rely on the observation that for $\gamma > 0$ constant, the feedback channel only needs to send $O(1)$ messages during the decoding of the first $(1-\gamma)k$ symbols. Using the erasure tolerance of the outer code, we can restrict the use of the real-time oblivious protocol to recover only $(1-\gamma)k$ symbols of the message encoded with the outer code, and the original message can then be recovered by using the decoding algorithm of the outer code.  

The second improvement comes from the non-uniform nature of the decoding cost. It is shown in \cite{beimel2007rt}, that each encoding symbol of the optimal degree, has at least a constant (in fact, $\ge 1/e$) chance of being decoded rather than discarded by the receiver. However, the cost of decoding the symbols is proportional to the encoding degree, which increases non-uniformly through the protocol and is much higher towards the end. By only requiring the receiver to decode $(1-\gamma)k$ symbols, we can avoid sending and processing many high encoding-degree symbols towards the end of the protocol, and achieve a truly linear decoding cost. Our main theorem is the following.

\begin{theorem} [Main]\label{thm:main}
Choose $k' \in \N$ and $\gamma > 0$ a constant. Let $\mathcal{R}_{k}$ be a linear erasure code of rate $1-2\gamma$ and block length $k$ which can be decoded on a BEC with erasure probability $\gamma$ in $O(k\log(1/\gamma))$ time with success probability $1-o_{k}(1)$. Then the code $(k',\mathcal{R}_{k},1-\gamma)$ can be real-time decoded in expected $O(k'\log(1/\gamma))$ time, using $(1+\gamma)2k' + O(1)$ encoding symbols and $O(1)$ feedback messages in expectation. The memory overhead for the encoder is $\frac{1}{1-2\gamma}k'$ and the decoding may fail with probability at most $o_{k}(1)$.
\end{theorem}

In section \ref{sec:trunc} we analyze a truncated version of the real-time oblivious erasure correction protocol of \cite{beimel2007rt}. In section \ref{sec:proof} we combine this with an appropriate outer code to complete the proof of theorem \ref{thm:main} before providing a small discussion in the last section.

\section{Analysis of the truncated real-time oblivious protocol of \cite{beimel2007rt}}\label{sec:trunc}

Fix $\gamma > 0$ a constant and use $k$ to denote the number of message symbols being sent using the real-time erasure correction protocol of \cite{beimel2007rt}. Here we analyze the $(1-\gamma)$-truncated version of their protocol, in which we use their algorithm to only recover only $(1-\gamma)k$ message symbols. This section proves the following theorem.

\begin{theorem}
Fix $\gamma >0$ a constant. In the real-time oblivious erasure correction of \cite{beimel2007rt}, the expected number of encoding symbols required to decode $(1-\gamma)k$ message symbols is $2k - \gamma ek + O(1)$. The decoding can be achieved using an expected $O(1)$ feedback messages, and the expected decoding complexity is $O(k\log(1/\gamma))$.
\end{theorem}

\begin{definition}
We use $p(d,r)$ to denote the revealing probability of a encoding symbol of degree $d$ when $r$ (out of $k$) message symbols are known to the receiver. That is, it is the probability (over the choice of symbols used in the encoding symbol), that an encoding symbol of degree $d$ will be successfully decoded by the receiver.

\end{definition}

\begin{definition}
We use $d(r)$ to denote the optimal degree for transmission when $r$ symbols are known to the receiver. For fixed $r$,

\[d(r) = \argmax_{1\le d \le k}(p(d,r))\]
\end{definition}

We recall the following facts.
\begin{lemma}[Lemma 4.5 from \cite{beimel2007rt}]
For $0 \le r < k-1$, \[d(r) = \left\lfloor\frac{k+1}{k-r}\right\rfloor\]
while for $r =k-1$, d(r) = k. 
\end{lemma}

\begin{lemma}[Lemma 4.8 from \cite{beimel2007rt}]\label{lem:1/e}
For $1 \le r \le k$, \[p(d(r),r) \ge \frac{1}{e}.\]
\end{lemma}

We now analyze truncated real-time oblivious erasure correction.
\subsection{Decoding complexity}
\begin{lemma}\label{lem:dec}
Fix $\gamma >0$ a constant. In order to recover a $(1-\gamma)$ fraction of the message symbols, the receiver needs to perform $O(k\log(1/\gamma))$ operations in expectation.  
\end{lemma}
\begin{proof}
First note that due to Lemma \ref{lem:1/e}, each encoding symbol of degree $d(r)$ is decoded with probability at least $1/e$, which means that in expectation, $e$ such encoding symbols  are sufficient to reveal one message symbol. Each time an encoding symbol is processed by the receiver, it needs to do work proportional to the degree of the encoding symbol to either discard or decode the encoding symbol. So expected computation for decoding $(1-\gamma)k$ symbols is

\begin{align*}
    \EE(\text{number of operations}) &= O\left(\sum_{r=0}^{(1-\gamma)k}\EE(r)d(r)\right)\\
    & = O\left(\sum_{r=0}^{(1-\gamma)k}e \left\lfloor\frac{k+1}{k-r}\right\rfloor\right)\\
    & = O\left(k\sum_{r=0}^{(1-\gamma)k} \frac{1}{k-r}\right)\\
     & = O\left(k\left(\sum_{i=1}^{k} \frac{1}{i} - \sum_{j=1}^{\gamma k - 1} \frac{1}{j}\right)\right)\\
    & = O\left(k\log(1/\gamma)\right)
\end{align*}
where the last equality uses the standard fact that $\log(k +1)\le \sum_{i=1}^{k} \frac{1}{i} \le \log(k) + 1$.
\end{proof}
\subsection{Feedback channel}

\begin{lemma}\label{lem:feedback}
Fix $\gamma >0$ a constant. In order to recover a $(1-\gamma)$ fraction of the message symbols, the feedback channel needs to successfully transmit at most $O(1)$ bits.   
\end{lemma}

\begin{proof}
The receiver sends a feedback message to update $d(r)$ every time the value of $d(r) = \left\lfloor\frac{k+1}{k-r}\right\rfloor$ changes. Note that $d(r)$ is monotone increasing as $r$ goes from $0$ to $k-1$ and that $d(0) = 1$. Then notice that

\begin{align*}
    d((1-\gamma)k) &= \left\lfloor\frac{k+1}{\gamma k} \right\rfloor \le \frac{2}{\gamma}
\end{align*}

So during the decoding of the first $(1-\gamma)k$ symbols, there are at most $2/\gamma$ values of $d(r)$, and the receiver needs to succesfully send at most $2/\gamma$ feedback messages.
\end{proof}

\subsection{Number of encoding symbols}
We upper bound the expected number of encoding symbols required for decoding $(1-\gamma)k$ symbols.

\begin{lemma}\label{lem:encodingsymbols}
Fix $\gamma >0$ a constant. The expected number of encoding symbols required to decode the first $(1-\gamma)k$ symbols is less than $2k - ke\gamma + O(1)$.
\end{lemma}
\begin{proof}
We lower bound the expected number of encoding symbols processed during the recovery of the first $(1-1/e)$ fraction of message symbols separately from the last $1/e$ fraction. In the case $\gamma \ge 1/e$, we note that $p(d(r),r) \ge p(1,r) = \frac{k-r}{k}$, and the expected number of symbols can be upper bounded by 

\begin{align*}
\sum_{r = 0}^{\lfloor k-\gamma k\rfloor}\frac{k}{k-r} &= k\sum_{i = \lfloor \gamma k\rfloor}^{k}\frac{1}{i}\\
&\le k(H_{k} - H_{\lfloor \gamma k\rfloor -1})\\
&\le k\left(\log(k) + \varphi + \frac{1}{2k} - \log\left(\frac{k}{e} - 2\right) - \varphi\right) \\
&\le k \left(\log\left(\frac{ek}{k - 2e}\right) + \frac{1}{2k}\right)\\
&\le k - k\log\left(1 - \frac{2e}{k}\right) +\frac{1}{2}\\
&\le k + 4e +\frac{1}{2}.
\end{align*}
Here, $H_k$ is the $k$th harmonic number. The second inequality uses the well-known fact that $\log(k) + \varphi \le H_k \le \log(k)+\varphi + \frac{1}{2k}$ where $\varphi$ is the Euler-Mascheroni constant. The final inequality uses the fact that $-\log(1-x) \le 2x$ for $0\le x \le \frac{1}{2}$.

Recall that the probability of decoding each symbol is bounded below by $1/e$ by Lemma \ref{lem:1/e}. Now assume that $\gamma < 1/e$. The expected cost is at most

\begin{align*}
\sum_{r = 0}^{\lfloor k- k/e\rfloor}\frac{k}{k-r} + \sum_{\lfloor k - k/e \rfloor + 1}^{\gamma k}e &\le k + ek\left(\frac{1}{e} - \gamma\right) + O(1)\\
&= 2k - \gamma ek + O(1).
\end{align*}

\end{proof}

\section{Proof of Theorem \ref{thm:main}}\label{sec:proof}

Now we choose the appropriate outer code and prove the claims of Theorem \ref{thm:main}, which we restate for convenience. 
 
 \begin{theorem*}
Choose $k' \in \N$ and $\gamma > 0$ a constant. Let $\mathcal{R}_{k}$ be a linear erasure code of rate $1-2\gamma$ and block length $k$ which can be decoded on a BEC with erasure probability $\gamma$ in $O(k\log(1/\gamma))$ time with success probability $1-o_{k}(1)$. Then the code $(k',\mathcal{R}_{k},1-\gamma)$ can be real-time decoded in expected $O(k'\log(1/\gamma))$ time, using $(1+\gamma)2k' + O(1)$ encoding symbols and $O(1)$ feedback messages in expectation. The memory overhead for the encoder is $\frac{1}{1-2\gamma}k'$ and the decoding may fail with probability at most $o_{k}(1)$.
\end{theorem*}
 
 Note that our outer code $\mathcal{R}_{k}$ has to satisfy the same properties as the outer codes of the well-known Raptor codes \cite{shokrollahi2006raptor}. In particular, $\mathcal{R}_{k}$ is a linear block code of block length $k$ which satisfies the following properties.  
 \begin{itemize}
 \item It has rate $R = 1-2\gamma$.
 \item It is decodable on a BEC with erasure probability $\frac{1-R}{2} = \gamma$ with success probability $1-o_{k}(1)$ using $O(k\log(1/\gamma))$ operations. 
 \end{itemize}
Several examples of appropriate codes are given in \cite{shokrollahi2006raptor}, for example, the codes of \cite{luby1997practical}. So we assume that $\gamma > 0$ is a fixed constant, and that for every $k$ we have linear code which satisfies the properties of $\mathcal{R}_{k}$ in Theorem \ref{thm:main}.

\begin{proof}[Proof of Theorem \ref{thm:main}] Recall that the code $(k',\mathcal{R}_{k},1-\gamma)$ first encodes the message $x$ of $k'$ symbols with $\mathcal{R}_{k}$, and then transmits $\mathcal{R}_{k}(x)$ over the BEC using the $(1-\gamma)$-truncated real-time oblivious decoding protocol. The decoder first recovers $(1-\gamma)k$ bits of $\mathcal{R}_{k}(x)$, and uses the decoding algorithm for $\mathcal{R}_{k}$ to recover $\mathcal{R}_{k}(x)$. Note that since $\mathcal{R}_{k}$ is linear, we can assume that it is systematic, and that $x$ can trivially be recovered from $\mathcal{R}_{k}(x)$. The decoding fails with probability $1-o_{k}(1)$, as this is the failure probability of the decoding algorithm of $\mathcal{R}_{k}$.

The claim about the expected number of feedback messages follows immediately from Lemma \ref{lem:feedback} since $k = \frac{k'}{1-2\gamma}$. From Lemma \ref{lem:encodingsymbols}, the expected number of encoding symbols is 

\[\frac{2k' - e\gamma k'}{1- 2\gamma} + O(1) \le (1+\gamma)2k' + O(1)\]

From Lemma \ref{lem:dec}, the expected decoding complexity is $O(k\log(1/\gamma))$ for the $(1-\gamma)$-truncated real-time oblivious decoding, followed by $O(k\log(1/\gamma))$ for decoding the outer code. Adding these together we get $O(k'\log(1/\gamma))$ as claimed. Finally, since the encoder needs to store $\mathcal{R}_{k}(x)$, we have a memory overhead of $\frac{1}{1-2\gamma}k'$ for the encoder.
\end{proof}

\section{Discussion}

\subsection*{Congestion avoidance}
One may also note that in the original protocol of \cite{beimel2007rt}, after $k-\sqrt{k+1}$ symbols have been recovered, the optimal encoding degree can change every time the receiver decodes a single encoding symbol. This can cause congestion on the feedback channel and messages can be sent with sub-optimal degree. While there are ways to remedy this (which increase the expected number of encoding symbols required for recovering the message), it is worth noting that our method avoids this problem entirely.  

\subsection*{Open problems}
As it stands, the receiver sends a message to the encoder to increase the encoding degree by one every time the value of $\left\lfloor \frac{k+1}{k-r}\right\rfloor$ changes. This maximizes the probability that each encoding symbol is decoded when it is processed by the receiver. Since we can reduce the expected cost of the feedback channel to $O(1)$, it would be interesting to also answer the following related question: given the ability to send a constant $t \in \N$ feedback messages, at what points should the receiver send messages to update the encoding degree (and by how much) as a function of $t$, to maximize the probability of decoding each encoding symbol?

\section{Acknowledgements}

The author would like to thank Mary Wootters, Moses Charikar and Li-Yang Tan for helpful comments and suggestions. This work was partially funded by NSF-CAREER grant CCF-1844628, NSF-BSF grant CCF-1814629, a Sloan Research Fellowship and a Google Graduate Fellowship.

\bibliographystyle{IEEEtran}
\bibliography{main}

\end{document}